\documentclass{bcp}
\usepackage{amsmath,amsthm}
\usepackage{amssymb}

\usepackage{graphicx}

\newtheorem{theorem}{Theorem}[section]
\newtheorem{proposition}[theorem]{Proposition}
\newtheorem{lemma}[theorem]{Lemma}
\newtheorem{corollary}[theorem]{Corollary}

\theoremstyle{definition}
\newtheorem{definition}[theorem]{Definition}
\newtheorem{example}[theorem]{Example}
\newtheorem{remark}[theorem]{Remark}

\usepackage[all]{xy}

\usepackage{hyperref}
\usepackage{amscd,euscript,array,mathrsfs}

\newenvironment{prp}{\begin{proposition}}{\end{proposition}}

\newenvironment{rmk}
{\begin{remark}}{\end{remark}}

\newenvironment{dfn}
{\begin{definition}}{\end{definition}}
\newenvironment{ex}
{\begin{example}}{\end{example}}

\newcommand{\R}{\mathbb R}

\newcommand{\C}{\mathbb C}

\newcommand{\N}{\mathbb N}
\newcommand{\bS}{\mathbb S}

\newcommand{\Z}{\mathbb Z}

\newcommand{\spn}{\mathrm{Span}}

\newcommand{\cD}{\mathcal D}

\newcommand{\cG}{\mathcal G}

\newcommand{\cS}{\mathcal S}

\newcommand{\ccH}{\mathscr H}
\newcommand{\ccB}{\mathscr B}

\newcommand{\g}{\mathfrak}
\newcommand{\ccK}{\mathscr K}

\newcommand{\lag}{\langle}
\newcommand{\rag}{\rangle}

\newcommand{\oline}{\overline}
\newcommand{\sseq}{\subseteq}

\newcommand{\Ad}{\mathrm{Ad}}

\newcommand{\dd}{\mathsf{d}}

\newcommand{\Lie}{\mathrm{Lie}}

\newcommand{\ood}{{\oline{\mathbf{1}}}}
\newcommand{\eev}{{\oline{\mathbf{0}}}}

\newcommand{\Hom}{\mathrm{Hom}}

\newcommand{\sfe}{\mathsf{e}}

\newcommand{\VecG}{\mathsf{Vec}_\Gamma}
\newcommand{\HH}{\mathscr H}
\newcommand{\Rep}{\mathsf{Rep}}
\newcommand{\semS}{\mathcal S}

\begin{document}

\mathclass{17B75; 22E45.}

\abbrevauthors{M. Mohammadi and H. Salmasian}
\abbrevtitle{THE GNS construction for $\Z_2^n$-graded Lie supergroups}

\title{The Gelfand--Naimark--Segal construction for unitary representations of $\Z_2^n$-graded Lie supergroups}

\author{Mohammad Mohammadi}
\address{Department of Mathematics,\\
Institute for Advanced Studies
in Basic Sciences (IASBS),\\
No. 444, Prof. Yousef Sobouti Blvd.
P. O. Box 45195-1159,\\
Zanjan, Iran\\
Email: moh.mohamady@iasbs.ac.ir}

\author{Hadi Salmasian}
\address{Department of Mathematics and Statistics,\\ University of Ottawa,
585 King Edward Ave,\\ Ottawa, ON,
Canada K1N 6N5.\\
Email: hsalmasi@uottawa.ca}

\maketitlebcp

\begin{abstract}
We establish a Gelfand--Naimark--Segal construction which yields a correspondence between cyclic unitary representations and positive definite superfunctions 
of a general class of 
$\Z_2^n$-graded Lie supergroups. 
\end{abstract}

\section{Introduction}

It is by now well known that unitary representations of Lie supergroups appear in various areas of mathematical physics related to supersymmetry
(see \cite{FQS}, \cite{GKO}, \cite{FSZ}, \cite{SaSt} as some important examples among numerous references). 
A mathematical approach to analysis on  unitary representations of Lie supergroups was pioneered in \cite{CCTV}, where unitary representations are defined  based on the notion of the Harish-Chandra pair associated to a Lie supergroup.

More recently, there is growing interest in studying generalized supergeometry, that is, geometry of graded manifolds where the grading group is not $\Z_2$, but $\Z_2^n:=\Z_2\times\cdots\times\Z_2$. The foundational aspects of the theory of $\Z_2^n$-supermanifolds were recently established in the works of 
Covolo--Grabowski--Poncin
(see \cite{CovoloPoncin2},  \cite{CovoloPoncin}) and 
Covolo--Kwok--Poncin 
(see \cite{CovoloPoncin3}).
 
In this short note, we make a first attempt to extend the theory of unitary representations to the
$\Z_2^n$-graded setting.
To this end, we use the concept of  a $\Z_2^n$-graded Harish-Chandra pair, which consists of a pair $(G_0,\g g)$ where $G_0$ is a Lie group and 
$\g g$ is a $\Z_2^n$-graded generalization of a Lie superalgebra, sometimes known as  
a \emph{Lie color algebra}.
Extending one of the main results of \cite{NeSaTG}, 
in Theorem \ref{GNS}
we prove that under a 	``perfectness'' condition on $\g g$, there exists a Gelfand--Naimark--Segal (GNS)  construction which yields a correspondence between positive definite smooth super-functions  and cyclic unitary representations of $(G_0,\g g)$. 
Theorem \ref{GNS} is applicable to interesting examples, such as 
$\Z_2^n$-Lie supergroups of classical type, e.g., the Harish-Chandra pair corresponding to $\g{gl}(V)$  defined in Example \ref{EXglV}, where $V$ is a $\Z_2^n$-graded vector space.

The key technical tool in the proof of Theorem \ref{GNS} is a Stability Theorem (see Theorem \ref{stability theorem}) which guarantees the existence of a unique unitary representation associated to a weaker structure, called a \emph{pre-representation}. Such a Stability Theorem holds unconditionally when $n=1$. But for $n>1$, it is not true in general. We are able to retrieve a variation of the Stability Theorem under the aforementioned extra condition on $\g g$. 
However, this still leaves the question of a general GNS construction open for further investigation.
We defer the latter question, presentation of explicit examples, as well as some proof details, to a future work. \\

\subsection*{Acknowledgements} We thank Professor Karl-Hermann Neeb for  illuminating conversations related to 
Theorem \ref{stability theorem} and Remark \ref{lifting}, and for many useful comments on a preliminary draft of this article, which improved our presentation substantially. This work was 
completed while the first author was visiting the University of Ottawa using a grant from the Iranian Ministry of Science, Research, and Technology. During this project, the second author was supported by an NSERC Discovery Grant.

\section{$\Z_2^n$-supergeometry}

We begin by reviewing the basic concepts of $\Z_2^n$-graded supergeometry, in the sense of 
\cite{CovoloPoncin}.
Let $\Gamma:=\Z_2^n:=\Z_2\times\cdots\times \Z_2$
where $\Z_2:=\{\eev,\ood\}$, and let $\mathsf b:\Gamma\times\Gamma\to \Z_2$ be a 
non-degenerate symmetric $\Z_2$-bilinear map. 
By a result of Albert (see \cite[Thm 6]{Albert} or \cite[Sec. 1--10]{Kaplansky}), if $\mathsf b(\cdot,\cdot)$ is of  alternate type (i.e., $\mathsf{b}\mathsf(a,a)=\eev$ for every $a\in\Gamma$), then $\mathsf b(\cdot,\cdot)$ is equivalent to the standard ``symplectic'' form
\[
\mathsf b_-(a,b):=\sum_{j=1}^n
a_{2j-1}b_{2j}+a_{2j}b_{2j-1}
\quad\text{ for every }a=(a_1,\ldots,a_n)\text{ and }b=(b_1,\ldots,b_n)\in\Gamma
,\]
whereas if $\mathsf b(a,a)\neq \eev$ for some $a\in \Gamma$, then $\mathsf b(\cdot,\cdot)$ is equivalent to the standard symmetric form \[
\mathsf b_+(a,b):=\sum_{j=1}^ma_jb_j
\quad\text{ for every }a=(a_1,\ldots,a_n)\text{ and }b=(b_1,\ldots,b_n)\in\Gamma
.
\] 

Henceforth we assume that $\mathsf b(\cdot,\cdot)$ is not of alternate type.
Since an equivalence of $\mathsf b(\cdot,\cdot)$ and $\mathsf b_+(\cdot,\cdot)$ is indeed an automorphism of the finite abelian 
group $\Gamma$, without loss of generality 
from now on we can assume that $\mathsf b=\mathsf b_+$.
In particular, from now on we represent an 
element $a\in \Gamma$ by $a:=\sum_{j=1}^n a_j\sfe_j$, where $\{\sfe_j\}_{j=1}^n$ is an orthonormal basis of $\Gamma$
with respect to $\mathsf b(\cdot,\cdot)$.

We now define $\beta:\Gamma\times \Gamma\to \{\pm1\}$ by  $\beta(a,b):=(-1)^{\mathsf b(a,b)}$ for  $a,b\in\Gamma$.
We equip the category $\VecG$ of $\Gamma$-graded complex vector spaces with the symmetry operator
\begin{equation}
\label{SVWW}
\bS_{V,W}:V\otimes W\to W\otimes V\ ,\ \bS_{V,W}(v\otimes w):=\beta(|v|,|w|)w\otimes v,
\end{equation}
where $|v|\in\Gamma$ denotes the degree of a homogeneous vector $v\in V$. 

\begin{rmk}
As it is customary in supergeometry, equality
\eqref{SVWW}
 should be construed as a relation for homogeneous vectors that  is subsequently extended by linearity to non-homogeneous vectors. 
In the rest of the manuscript 
 we will stick to this convention.
\end{rmk}

Equipped with $\bS$ and the usual 
$\Gamma$-graded tensor product of vector spaces, 
 $\VecG$ is a symmetric  monoidal category.
This fact was also observed in \cite[Prop. 1.5]{Benk}.

\begin{rmk}
\label{lifting}
Note that $\beta:\Gamma\times \Gamma\to\{\pm1\}$ is a 2-cocycle. In fact it represents the obstruction to the lifting of the map $\Gamma\to\{\pm 1\}$, $(a_1,\dots,a_n)\mapsto (-1)^{\sum_{i=1}^na_i}$, with respect to the exact sequence 
\[
\xymatrix{
1\ar[r] & \{\pm 1\}\ar[r] & \{\pm 1,\pm i\}\ar[r]^{\ \ \,t\mapsto t^2}& \{\pm 1\} \ar[r] &1},
\]
where $i:=\sqrt{-1}$. More precisely, the lifting obstruction cocycle is naturally represented by 
\[
\delta:\Gamma\times \Gamma\to \{\pm 1\}\ ,\ (a_1,\ldots,a_n)\mapsto (-1)^{\sum_{1\leq i,j\leq n}a_i b_j},
\]
but $\beta-\delta=d\eta$ where $\eta(a_1,\ldots,a_n):=(-1)^{\sum_{1\leq i\neq j\leq n}a_ia_j}$. We thank Professor  Karl-Hermann Neeb for letting us know about this property of $\beta$.
\end{rmk}

Let $\prec$ denote the lexicographic order on elements of $\Gamma$. That is, for $a:=\sum_{j=1}^n a_j\sfe_j$ and $b:=\sum_{j=1}^nb_j\sfe_j$, we set $a\prec b$ if and only if there exists some $1\leq j\leq n$ such that
$a_j=\eev$, $b_j=\ood$, and 
$a_k=b_k$ for all $k<j$, . 
Thus we can express $\Gamma$ as $\Gamma=\{\gamma_0\prec \gamma_1\prec\cdots\prec \gamma_{2^n-1}\}$, where $\gamma_0=0$. 

Let $p\in \N\cup\{0\}$ and let ${\mathbf q}:=(q_1,\ldots,q_{2^n-1})$ be a $(2^n-1)$-tuple such that $q_j\in \N\cup\{0\}$ for all $j$. We set 
 $|{\mathbf q}|:=\sum_{j=1}^{2^n-1}q_j$. 
By a \emph{$\Gamma$-superdomain} of dimension $p|\mathbf{q}$ we mean a locally ringed space $(U,\mathcal O_U)$ such that $U\subset\R^p$ is an open set, and the structure sheaf $\mathcal O_U$ is given by 
\[
\mathcal O_U(U'):=C^\infty(U';\C)[[\xi_1,\cdots,\xi_{|\mathbf q|}]]\quad\text{ for every open set }U'\sseq U,\]
where the right hand side denotes the $\Gamma$-graded algebra of formal power series 
in
variables $\xi_1,\ldots,\xi_{|{\mathbf{q}}|}$
with coefficients in $C^\infty(U';\C)$,
 such that $|\xi_{r}|=\gamma_k$ for 
$\sum_{j=1}^{k-1}q_j<r\leq \sum_{j=1}^{k}q_j$, subject to relations
\[
\xi_r\xi_s=\beta(|\xi_r|,|\xi_s|)\xi_s\xi_r
\quad\text{ for every }1\leq r,s\leq |\mathbf q|.
\]
By a \emph{smooth $\Gamma$-supermanifold} $M$ of dimension $p|\mathbf q$ we mean a locally ringed space $(M, \mathcal O_M)$ that is locally isomorphic to a $p|\mathbf q$-dimensional $\Gamma$-superdomain. From this viewpoint, a \emph{$\Gamma$-Lie supergroup} is a group object in the category of smooth $\Gamma$-supermanifolds.

As one expects, to a $\Gamma$-Lie supergroup one can canonically associate a \emph{$\Gamma$-Lie superalgebra}, which is an object of $\VecG$ of the form $\mathfrak g=\bigoplus_{a\in \Gamma}\mathfrak g_a$, equipped with a $\Gamma$-superbracket
$[\cdot,\cdot]: \mathfrak g\times\mathfrak g \to\g g$
that satisfies the following properties:
 \begin{enumerate}
\item [(i)] $[\cdot,\cdot]$ is bilinear, and $[\g g_a,\g g_b]
\subset \g g_{ab}$ for  $a,b\in\Gamma$.
\item[(ii)]    $[x,y]=-\beta(a,b)[y,x]$ for  $x\in\mathfrak g_a, y\in\mathfrak g_b$, where $a,b\in\Gamma$.
\item [(iii)] $[x,[y,z]]=[[x,y],z]+\beta(a,b)\beta(a,c)[y,[z,x]]$ for  $x\in \mathfrak g_a, y\in\mathfrak g_b, z\in\mathfrak g_c$, where $a,b,c\in\Gamma$.
  \end{enumerate}
\begin{rmk}
We remark that a $\Gamma$-Lie superalgebra is more commonly known as a \emph{Lie color algebra}. Nevertheless, in order to keep our nomenclature  compatible with  \cite{CovoloPoncin}, we
use the term  $\Gamma$-Lie superalgebra instead.
\end{rmk}

From classical supergeometry (that is, when $n=1$) one knows that the category of Lie supergroups can be replaced by another category with a more concrete structure,  known as the category of \emph{Harish-Chandra pairs}  (see \cite{Kostant}, \cite{Koszul}). A similar statement holds in the case of $\Gamma$-Lie supergroups.
\begin{dfn}
A \emph{$\Gamma$-Harish-Chandra pair} is a pair $(G_0,\g g)$ where $G_0$ is a Lie group, and $\mathfrak g=\bigoplus_{a\in \Gamma}\mathfrak{g}_a$ is a $\Gamma$-Lie superalgebra equipped with an action  $\Ad:G_0\times \g g\to \g g$ of $G_0$ by linear operators that preserve the $\Gamma$-grading, 
which extends the adjoint action of $G_0$ on $\g g_0\cong\Lie(G_0)$.  
\end{dfn}

\begin{ex}
\label{EXglV}
Let $V:=\bigoplus_{a\in\Gamma}V_a$ be a $\Gamma$-graded vector space. The $\Gamma$-Lie superalgebra $\g{gl}(V)$ is the vector space of linear transformations on $V$, with superbracket $[S,T]:=ST-\beta(|S|,|T|)TS$. One can also consider the $\Gamma$-Harish-Chandra pair $(G_0,\g g)$, where $G_0\cong\prod_{a\in\Gamma}\mathrm{GL}(V_a)$.

\end{ex}

Indeed the $\Gamma$-Harish-Chandra pairs  form a category in a natural way.  The morphisms of this category are pairs of maps $
(\phi,\varphi):(G,\g g)\to (H,\g h),
$ where $\phi:G\to H$ is a Lie group homomorphism and $\varphi:\g g\to\g h$ is a morphism in the category of $\Gamma$-Lie superalgebras such that $\varphi\big|_{\g g_0}=\dd\phi$.
The following statement plays a key role in the study of $\Gamma$-Lie supergroups.
\begin{prp}
The category of $\Gamma$-Lie supergroups is isomorphic to the category of $\Gamma$-Harish-Chandra pairs.
\end{prp}
\begin{proof}
The proof is a straightforward but lengthy modification of the argument for the analogous result in the case of ordinary supergeometry (see   \cite{Kostant}, \cite{Koszul}, or \cite{Fioresi}). Therefore we only sketch an outline of the proof. The functor from  $\Gamma$-Lie supergroups to $\Gamma$-Harish-Chandra pairs is easy to describe: a $\Gamma$-Lie supergroup $(G_0,\mathcal O_{G_0})$ is associated to the Harish-Chandra pair $(G_0,\g g)$, where $\g g$ is the $\Gamma$-Lie superalgebra of $(G_0,\mathcal O_{G_0})$. 
As in the classical super case, the functor associates a homomorphism of $\Gamma$-Lie supergroups  $(G_0,\mathcal O_{G_0})\to(H_0,\mathcal O_{H_0})$ to the pair of underlying maps $G_0\to H_0$ and the tangent map at identity $\g g\to\g h$.
Conversely, from a $\Gamma$-Harish-Chandra pair $(G_0,\g g)$ we construct a Lie supergroup $(G_0,\mathcal O_{G_0})$ as follows. For every open set $U\sseq G_0$, we set
$\oline{\mathcal O}_{G_0}(U):=\Hom_{\g g_0}
\big(\mathfrak{U}(\g g_\C^{}),C^\infty(U;\mathbb C)\big)
$, where $\mathfrak{U}(\g g_\C^{})$ denotes the universal enveloping algebra of $\g g_\C^{}:=\g g\otimes_\R\C$. The $\Gamma$-superalgebra structure on $\oline{\mathcal O}_{G_0}(U)$ is defined using the algebra structure of 
$C^\infty(U;\mathbb C)$ and the $\Gamma$-coalgebra structure of $\mathfrak{U}(\g g_\C^{})$, 
exactly as in the classical super case.
Using the PBW Theorem for $\Gamma$-Lie superalgebras (see \cite{Scheunert}) one can see that $(G_0,\oline{\mathcal O}_{G_0})$ is indeed a  
$\Gamma$-supermanifold (see \cite[Prop. 7.4.9]{Fioresi}).
The definition of the $\Gamma$-Lie supergroup structure of 
$(G_0,\oline{\mathcal O}_{G_0})$ is similar to the classical case as well (see \cite[Prop. 7.4.10]{Fioresi}). Finally, to show that the two functors are inverse to each other, we need a  sheaf isomorphism $\mathcal O_{G_0}\cong
\oline{\mathcal O}_{G_0}$. This sheaf isomorphism is given by the maps
\[
\mathcal O_{G_0}(U)\to
\oline{\mathcal O}_{G_0}(U)\ ,\
s\mapsto\left[D\mapsto \beta(|D|,|s|)\widetilde{\mathrm{L}_D s}\right]
\text{ for }D\in\mathfrak{U}(\g g_\C^{}),\] 
where $|D|,|s|\in\Gamma$ are the naturally defined degrees, $\mathrm{L}_D$ denotes the left invariant differential operator 
on $(G_0,\mathcal O_{G_0})$
corresponding to $D$, and $\widetilde{\mathrm{L}_D s}$ means evaluation of the section at points of $U$. The proof of bijective correspondence of morphisms in the two categoies is similar to 
\cite[Prop. 7.4.12]{Fioresi}.
\end{proof}

\section{$\Gamma$-Hilbert superspaces and unitary representations}
In order to define a unitary representation of a $\Gamma$-Harish-Chandra pair, one needs to obtain a well-behaved definition
 of Hilbert spaces in the category $\VecG$. 
This is our first goal in this section.

For any $a=\sum_{j=1}^na_j\sfe_j\in \Gamma$, set 
$\mathbf u(a):=|\{1\leq j\leq n\,:\,a_j=\ood\}|$ (for example, $\mathbf u(\sfe_1+\sfe_3+\sfe_4)=3$) and define
  \begin{align}\label{alpha0}
  \alpha(a):=e^{\frac{\pi i}{2}\mathbf u(a)}.
  \end{align}

  \begin{dfn}
  \label{dfn-prehilb}
  Let $\HH\in\mathsf{Obj}(\VecG)$. We call $\HH$ a \emph{$\Gamma$-inner product space} if and only if it is equipped with a non-degenerate sesquilinear form $\langle\cdot,\cdot\rangle$ that satisfies the following properties:
  \begin{enumerate}
  \item[i)] $\langle\HH_a,\HH_b\rangle=0$, for $a,b\in \Gamma$ such that $a\neq b$.
  \item[ii)] $\langle w,v\rangle=\beta(a,a)\overline{\langle v,w\rangle}$, for $v,w\in\mathscr H_a$ where $a\in\Gamma$.
  \item[iii)]  $\alpha(a)\langle v,v\rangle\geq 0$, for every $v\in \mathscr H_a$ where $a\in \Gamma$.
  \end{enumerate} 
  \end{dfn}
  \begin{rmk}
  The choice of $\alpha$ in \eqref{alpha0} is made as follows.
 The most natural property that one expects from Hilbert spaces is that the tensor product of (pre-)Hilbert spaces is a (pre-)Hilbert space.
     Thus, we are seeking $\alpha:\Gamma\to\C^\times$ such that the tensor product of two 
  $\Gamma$-inner product spaces is also a $\Gamma$-inner product space. 
Clearly after scaling the values of $\alpha$ by positive real numbers we can assume that $\alpha(\Gamma)\sseq\{\pm 1,\pm i\}$.   
  For two
  $\Gamma$-inner product spaces
  $\mathscr H$ and $\mathscr  K$, one has
  \[
  (\mathscr H\otimes \mathscr K)_a:=\bigoplus_{bc=a}
  \mathscr H_{b}\otimes\mathscr H_{c},
  \]
  equipped with the canonically induced sesquilinear form
  \[
  \langle v\otimes w,v'\otimes w'\rangle_{\mathscr H\otimes\mathscr K}:=\beta(|w|,|v'|)\langle v,v'\rangle_\mathscr H\langle w,w'\rangle_\mathscr K.
  \]
Closedness under tensor product implies that
 \begin{equation}
 \label{EqDfalphb}
  \alpha(bc)=\beta(b,c)\alpha(b)\alpha(c),\qquad \text{for all } b,c \in \Gamma.
  \end{equation}
In the language of group cohomology, this means that the 2-cocycle $\beta$ satisfies $\beta=d\alpha$. Consequently,
up to twisting by a group homomorphism $\Gamma\to \{\pm 1\}$, there exists a unique $\alpha$ which satisfies the latter relation.
  \end{rmk}

  \begin{rmk}
  \label{rmk<>from()}
  Associated to any $\Gamma$-inner product  space $\HH$, there is an inner product (in the ordinary sense) defined by
  \[
  (v,w) :=\begin{cases}
  0& \text{ if }|v|\neq |w|,\\
  \alpha(a)\langle v,w\rangle& \text{ if }|v|=|w|=a\text{ where }a\in\Gamma.
  \end{cases}
  \]
  The vector space $\HH$, equipped with $(\cdot,\cdot)$, is indeed a pre-Hilbert space in the usual sense. Thus we can consider the completion of $\ccH$ with respect to the norm $\|v\|:=(v,v)^{\frac{1}{2}}$. By reversing the process of obtaining $(\cdot,\cdot)$ from $\lag\cdot,\cdot\rag$, we obtain a $\Gamma$-inner product on the completion of $\ccH$.

\end{rmk}

  \begin{dfn}
Let $\HH$ be a $\Gamma$-inner product space, and let 
$T\in\mathrm{End}_\C(\HH)$. 
We define the \emph{adjoint} $T^\dagger$ of $T$  
as follows. If 
$T\in\mathrm{End}_\C(\HH)_a$ for some $a\in\Gamma$, we 
define $T^\dagger$ by
  \[
  \langle v,Tw\rangle=\beta(a,b)\langle T^\dagger v,w\rangle,
  \]
  where $v\in\mathscr H_b$.  We then extend the assignment $T\mapsto T^\dagger$ to a conjugate-linear map on $\mathrm{End}_\C(\HH)$.
  \end{dfn}

Now let $\ccH$ be a $\Gamma$-inner product space, and let
$(\cdot,\cdot)$ be the ordinary 
inner product associated to $\ccH$ as in Remark \ref{rmk<>from()}. Then
for a linear map $T:\HH\to \HH$ 
we can define an adjoint 
with respect to $(\cdot,\cdot)$, by the relation
  \[
  (Tv,w)=(v,T^*w)\qquad\text{ for every }v,w\in\HH.
  \] 
A straightforward calculation yields  
\[
  T^*=\overline{\alpha(|T|)}T^\dagger.
  \]  
It follows that $T^{\dagger \dagger}=T$. Furthermore, $(ST)^\dagger=\beta(|S|,|T|)T^\dagger S^\dagger$.

We are now ready to define unitary representations of $\Gamma$-Harish-Chandra pairs. 
The definition of a unitary representation of a
$\Gamma$-Harish-Chandra pair is a natural extension of the one for the $\Z_2$-graded case. First recall that  
for a unitary representation $(\pi,\ccH)$ of a Lie group $G$ on a Hilbert space $\ccH$, we denote the space of $C^\infty$ vectors by $\ccH^\infty$. Thus 
the vector space $\ccH^\infty$ consists of all vectors $v\in \ccH$ for which the map $G\to \ccH$, $g\mapsto \pi(g)v$ is smooth. Furthermore, given $x\in\Lie(G)$ and $v\in\HH$, we set \[
\dd\pi(x):=\lim_{t\to 0}\frac{1}{t}\left(
\pi(\exp(tx))v-v\right),\]
whenever the limit exists. We denote the domain of the unbounded operator $\dd\pi(x)$ by $\cD(\dd\pi(x))$. The unbounded operator $-i\dd\pi(x)$ is the self-adjoint generator of  the one-parameter unitary representation $t\mapsto \pi(\exp(tx))$. For a comprehensive exposition of the theory of unitary representations and relevant facts from the theory of unbounded operators, see \cite{Warner}.

\begin{dfn}
\label{urep}
A \emph{smooth unitary representation} of  a 
$\Gamma$-Harish-Chandra pair $(G_0,\mathfrak g)$ is a triple $(\pi,\rho^\pi,\HH)$ that satisfies the following properties:

\begin{enumerate}

\item [(R1)] $(\pi,\HH)$ is a smooth unitary representation of the Lie group $G_0$ on the $\Gamma$-graded Hilbert space $\HH$
by operators $\pi(g)$, $g\in G_0$, which preserve the $\Gamma$-grading.

\item [(R2)] $\rho^\pi:\mathfrak g\to \mathrm{End}_{\mathbb C}(\mathscr H^\infty)$ is a representation of the $\Gamma$-Lie superalgebra $\mathfrak g$.

\item [(R3)] $\rho^\pi(x)=\dd\pi(x)\big|_{\mathscr H^\infty}$ for  $x\in \mathfrak g_{
   0}$.

\item [(R4)] $\rho^\pi(x)^\dagger=-\rho^\pi(x)$
    for  $x\in \mathfrak g$.
        
\item [(R5)]  $\pi(g)\rho^\pi(x)\pi(g)^{-1} = \rho^\pi(\mathrm{Ad}(g)x)$ for $g\in G_0$ and $x\in \g g$.

\end{enumerate}
\end{dfn}
Unitary representations of a $\Gamma$-Harish-Chandra pair form a category
$\Rep=\Rep(G_0,\g g)$.
 A morphism in this category from $(\pi,\rho^\pi,\ccH)$ to $(\sigma,\rho^\sigma,\ccK)$ is a bounded linear map $T:\ccH\to\ccK$ which respects the $\Gamma$-grading and satisfies $T\pi(g)=\sigma(g)T$ for $g\in G_0$ (from which it follows  that $T\ccH^\infty\sseq \ccK^\infty$) and
    $T\rho^\pi(x)=\rho^\sigma(x)T$ for $x\in\g g$.
\begin{rmk}
At first glance, it seems that the definition of a unitary representation of a $\Gamma$-Harish-Chandra pair depends on the choice of $\alpha$. Nevertheless, it is not difficult to verify that 
for two coboundaries $\alpha,\alpha'$ which satisfy \eqref{EqDfalphb}, the corresponding categories
$\Rep_\alpha$ and $\Rep_{\alpha'}$ are isomorphic. Indeed if $\chi:\Gamma\to \{\pm1\}$ is a group homomorphism such that $\alpha'=\chi\alpha$, then 
we can define a functor $\mathscr F:\Rep_\alpha\to\Rep_{\alpha'}$ which maps 
$(\pi,\rho^\pi,\HH)$ to $(\pi',\rho^{\pi'},\HH')$ where $\HH:=\HH'$
(but the $\Gamma$-inner product of $\HH'$ is defined by
$
\lag v,v \rag_{\HH'}=\chi(a)\lag v,v\rag_{\HH}
$ for $v\in\ccH_a$ where $a\in\Gamma$), $\pi':=\pi$, and $\rho^{\pi'}(x):=\chi(a)\rho^\pi(x)$ for $x\in\g g_a$ and $a\in\Gamma$. The functor
$\mathscr F$ is defined to be identity on morphisms. That is, 
for a morphism $T:(\pi,\rho^\pi,\ccH)\to (\sigma,\rho^\sigma,\ccK)$ we set $\mathscr F(T):=T$. It is straightforward to see that
with 
$(\pi',\rho^{\pi'},\ccH')$ and $(\sigma',\rho^{\sigma'},\ccK')$ defined as above, the map
$T:(\pi',\rho^{\pi'},\ccH')\to(\sigma',\rho^{\sigma'},\ccK')$ is still an intertwining map.
The main point is that $\rho^{\pi'}$ and 
$\rho^{\sigma'}$ are obtained from $\rho^\pi$ and $\rho^\sigma$ via scaling by the same scalar. The inverse of $\mathscr F$ is defined similarly.
\end{rmk}

\section{The stability theorem}
We now proceed towards the statment and proof of the Stability Theorem. 
In what follows, we will need the following technical definition (see \cite{NeSaMZ}).
\begin{dfn}\label{prerep def}
Let $(G_0,\mathfrak g)$ be a $\Gamma$-Harish-Chandra pair. By a \emph{pre-representation} of $(G_0,\mathfrak g)$, we mean a 4-tuple $(\pi,\ccH,\ccB,\rho^{\ccB})$ that satisfies the following conditions:
\begin{enumerate}

\item [(PR1)] $(\pi,\HH)$ is a smooth unitary representation of the Lie group $G_0$ on the $\Gamma$-graded Hilbert space $\HH$
by operators $\pi(g)$, $g\in G_0$, which preserve the $\Gamma$-grading.

\item [(PR2)] $\ccB$ is a dense, $G_0$-invariant, and $\Gamma$-graded subspace of $\mathscr H$ that is contained in $\bigcap_{x\in \mathfrak g_0}\cD({\dd\pi}(x))$.
     \item [(PR3)] $\rho^\ccB:\mathfrak g\to \mathrm{End}_{\mathbb C}(\ccB)$ is a representation of the $\Gamma$-Lie superalgebra $\mathfrak g$.
    \item [(PR4)] $\rho^\ccB(x)={\dd\pi}(x)\big|_{\ccB}$ and $\rho^\ccB(x)$ is essentially skew adjoint for $x\in \mathfrak g_{0}$.
    \item [(PR5)] $\rho^\ccB(x)^\dagger=-\rho^\ccB(x)$
for $x\in \mathfrak g$.
    \item [(PR6)] $\pi(g)\rho^\ccB(x)\pi(g)^{-1} = \rho^\ccB(\mathrm{Ad}(g)x)$ for $g\in G_0$ and $x\in \g g$.    \end{enumerate}
\end{dfn} 
\begin{rmk}
It is shown in \cite[Rem. 6.5]{NeSaMZ} that 
(PR2) and (PR3) imply that $\ccB\sseq\HH^\infty$. 
\end{rmk}
Set 
\[
\Gamma_\eev;=\{a\in\Gamma\,:\,\mathsf{b}(a,a)=\eev\}\text{ and }\Gamma_\ood:=
\{a\in\Gamma\,:\,\mathsf{b}(a,a)=\ood\}
\]
\begin{theorem}
\label{stability theorem}
{\upshape \textbf{(Stability Theorem)}}
Let $(\pi,\ccH,\ccB,\rho^\ccB)$ be a pre-representation of a $\Gamma$-Harish-Chandra pair $(G_0,\mathfrak g)$. Assume that for every $a\in\Gamma_\eev\backslash\{0\}$ we have \[
\mathfrak g_{a}=\sum_{b,c\in\Gamma_\ood,bc=a}[\mathfrak g_{b},\mathfrak g_{c}].
\]
 Then there exists a unique extension of $\rho^\ccB$ to a linear map 
$\rho^\pi:\mathfrak g\rightarrow \mathrm{End}_{\mathbb C}(\mathscr H^\infty)$ such that $(\pi,\rho^\pi,\mathscr H)$ is a smooth unitary representation of $(G_0,\mathfrak g)$.
\end{theorem}
\begin{proof}
For every $a\in\Gamma_\ood$ the direct sum $\g g_0\oplus\g g_a$ is a Lie superalgebra (in the $\Z_2$-graded sense). We define a $\Z_2$-grading
of $\HH$ by
\[
\ccH_\eev:=\bigoplus_{\mathsf{b}(a,b)=\eev}\HH_b\ \text{ and }
\ \ccH_\ood:=\bigoplus_{\mathsf{b}(a,b)=\ood}\HH_b
,\]
and then we use the Stability Theorem in the $\Z_2$-graded case (see \cite[Thm 6.14]{NeSaMZ})
to obtain that 
there exists a unique $\Gamma$-Lie superalgebra homomorphism $\rho^{\pi,a}:\mathfrak g_{0}\oplus\mathfrak g_a\to \mathrm{End}_{\mathbb C}(\HH^\infty)$, such that $(\pi,\rho^{\pi,a},\ccH)$ is a unitary representation of the ($\Z_2$-graded) Harish-Chandra pair $(G_0,\g g_0\oplus\g g_a)$. Since the maps $\rho^{\pi,a}$ agree on $\g g_0$, they give rise to a linear map
\[
\mathfrak g_{0}\oplus(\bigoplus_{a\in\Gamma_\ood}\mathfrak g_{a})\to \mathrm{End}_{\mathbb C}(\HH^\infty).
\]
It remains to obtain a suitable extension of the latter map to $\mathfrak g_a$ for every $a\in\Gamma_\eev\backslash\{0\}$.

Fix $x\in\mathfrak g_{a}$, $a\in\Gamma_\eev\backslash\{ 0\}$. Then
we can write $x=\sum_j[y_j,z_j]$
such that $y_j\in \g g_{b_j}$, $z_j\in\g g_{c_j}$, where $b_j,c_j\in\Gamma_\ood$, and 
$b_jc_j=a$. 
For every $v\in \ccH^\infty$, we define  
\begin{equation}
\label{rhoPIDFN}
\rho^\pi(x)v:=
\sum_j\big(\rho^{\pi,b_j}(y_j)\rho^{\pi,c_j}(z_j)v-\beta(b_j,c_j)\rho^{\pi,c_j}(z_j)\rho^{\pi,b_j}(y_j)v\big)
.\end{equation}
Let us first verify that $\rho^\pi(x)v$ is well-defined. To this end, it suffices to show that if  $\sum_j[y_j,z_j]=0$, then the right hand side of 
\eqref{rhoPIDFN} vanishes. 
To verify the latter assertion, observe that for every 
$w\in\ccB_s$, $s\in\Gamma$, we have
\begin{align*}
\langle w,\rho^\pi(x)v \rangle&=\sum_j
\beta(b_j,s)\beta(c_j,sb_j)\langle\rho^{\pi,c_j}(z_j)^\dagger\rho^{\pi,b_j}(y_j)^\dagger w,v\rangle\\
&-\sum_j
\beta(c_j,s)\beta(b_j,sc_j)\beta(b_j,c_j)
\langle
\rho^{\pi,b_j}(y_j)^\dagger\rho^{\pi,c_j}(z_j)^\dagger w,v\rangle\\
&=\beta(a,s)\left\langle
\sum_j\beta(b_j,c_j)\rho^{\pi,c_j}(z_j)^\dagger\rho^{\pi,b_j}(y_j)^\dagger w-\sum_j
\rho^{\pi,b_j}(y_j)^\dagger\rho^{\pi,c_j}(b_j)^\dagger w
\,,v\right\rangle\\
&=-\beta(a,s)\left\langle \sum_j \rho^\ccB([y_j,z_j])w,v\right\rangle =0.
\end{align*}
The assertion now follows from density of $\ccB$ in $\HH$.

It remains to show that given any $x_a\in\g g_a$ and $y_b\in\g g_b$, where $a,b\in\Gamma$,  we have
\begin{align}\label{abc}
\rho^\pi([x_a,y_b])v=\big(\rho^\pi(x_a)\rho^\pi(y_b)
-\beta(a,b)\rho^\pi(y_b)\rho^\pi(x_a)\big)v
\quad\text{ for }v\in\HH^\infty
.
\end{align}
To verify the last equality note that for every  
$w\in\ccB_s$ where $s\in\Gamma$, 
\begin{align*}
\langle
w,\rho^\pi(x_a)&\rho^\pi(y_b)v
-\beta(a,b)\rho^\pi(y_b)\rho^\pi(x_a)v 
\rangle\\
&
=\beta(a,s)\beta(as,b)\langle \rho^\pi(y_b)^\dagger\rho^\pi(x_a)^\dagger w,v
\rangle
-
\beta(b,s)\beta(s,a)
\langle 
\rho^\pi(x_a)^\dagger\rho^\pi(y_b)^\dagger w,v
\rangle\\
&=-\beta(a,s)\beta(b,s)
\Big\langle 
\rho^\ccB(x_a)\rho^\ccB(y_b)w-\beta(a,b)
\rho^\ccB(y_b)\rho^\ccB(x_a)w,v
\Big\rangle\\
&=
\beta(a,s)\beta(b,s)
\Big\langle
\rho^\ccB\big([x_a,y_b]\big)^\dagger w,v
\Big\rangle=
\Big\langle
w,\rho^\pi\big([x_a,y_b]\big) v
\Big\rangle
.
\end{align*}
Again density of $\ccB$ in $\HH$ implies that both sides of \eqref{abc} are equal. Finally, uniqueness of $\rho^\pi(x)$ can be proved by the same technique.
\end{proof}

\begin{ex}
\label{badexample}
For $n>1$, one cannot expect the Stability Theorem to hold without any condition on $\g g$. For example, let  $\Gamma=\Z_2\times \Z_2$, let $G_0$ be the trivial group,  and let $\g g=\g g_{\eev\eev}\oplus
\g g_{\eev\ood}
\oplus
\g g_{\ood\eev}
\oplus 
\g g_{\ood\ood}
$ where 
\[
\g g_{\ood\ood}:=\R\text{ and }
\g g_{\eev\eev}:=
\g g_{\eev\ood}:=
\g g_{\ood\eev}:=\{0\}.
\]
Then a pre-representation of $(G_0,\g g)$ is the same  (up to scaling)  as a symmetric operator defined on a dense subspace of a Hilbert space, whereas a unitary representation of $(G_0,\g g)$ is the same 
 (up to scaling) 
as a \emph{bounded} self-adjoint operator. Therefore a pre-representation does not necessarily extend to a unitary representation.   
\end{ex}

\section{The GNS representation}

Our goal in this section is to extend the GNS construction of \cite{NeSaTG} to the setting of $\Gamma$-Harish-Chandra pairs. We begin by outlining some generalities about this construction. 
 Let $(G_0,\g g)$ be a $\Gamma$-Harish-Chandra pair, and let
$\cG:=(G_0,\mathcal O_{G_0})$ denote the $\Gamma$-Lie supergroup corresponding to  $(G_0,\g g)$. 
 Our main goal, to be established in 
Theorem \ref{GNS},  
  is to construct a correspondence between unitary representations of 
$(G_0,\g g)$ with a cyclic vector, and smooth positive definite functions on 
$\cG$. 
The suitable definition of a unitary representation with a cyclic vector is given in 
Definition \ref{cyclicvcc}. The main subtlety
is to define positive definite functions on $\cG$. To this end, we use the method developed in
\cite{NeSaTG} for the classical super case, which we describe below.
Recall that the $\Gamma$-superalgebra
$C^\infty(\cG)$
 of smooth functions on
$\cG$  is isomorphic to
\[\Hom_{\g g_{0}}
\big(\mathfrak{U}(\g g_\C),C^\infty(G_0;\mathbb C)\big),
\]
where $\g g_0$ acts on $C^\infty(G_0;\C)$ by left invariant differential operators. 
We realize elements of the latter algebra as functions on a semigroup $\cS$ which is equipped with an involution. Then we use the abstract definition of a positive definite function on a semigroup with an involution (see Definition
\ref{dfn-posddd}). 

Given a unitary representation of $(G_0,\g g)$ with a cyclic vector, the matrix coefficient of that
cyclic vector will be a positive definite element of $C^\infty(\cG)$. Conversely, from a positive definite $f\in C^\infty(\cG)$,
we construct a unitary representation 
of $(G_0,\g g)$ 
by considering the reproducing kernel Hilbert space associated to $f$, now realized as a function on $\cS$. The Lie group $G_0$ has a canonical action  
on the reproducing kernel Hilbert space.
However, the $\Gamma$-Lie superalgebra $\g g$
acts on a dense subspace of the latter Hilbert space, which is in general strictly smaller than the space of smooth vectors for the $G_0$-action. 
 The main part of the proof of Theorem \ref{GNS} is to show that  the action of  $\g g$ is well defined on the entire space of smooth vectors. 
For this we need 
 Theorem \ref{stability theorem}, 
 which is the reason for presence of 
 the condition 
\eqref{eq-condgagbb} in Theorem \ref{GNS}.

Set $\g g_\C^{}:=\g g\otimes_\R\C$. Let $\mathfrak{U}(\g g_\C^{})$ be the universal enveloping algebra of $\g g_\C^{}$, that is,  the quotient $T(\g g_\C^{}) / I$, where $T(\g g_\C^{})$ denotes the tensor algebra of $\g g_\C^{}$ in the category $\VecG$, and 
$I$ denotes the two-sided   ideal of  $T(\g g_\C^{})$ generated by elements of the form
\[
x\otimes y-\beta(|x|,|y|)y\otimes x-[x,y],\quad\text{ for homogeneous }x,y\in\g g_\C^{}.
\]
See \cite[Sec. 2.1]{Nishiyama90} for further details.

Let $x \mapsto x^*$ be the (unique) conjugate-linear map 
on $\g g_\C^{}$ that is defined by the relation   
\begin{align}\label{antilinear map}
              x^*:=-\overline{\alpha(a)}x,\qquad \text{ for every }x\in \mathfrak g_a,\ a\in\Gamma.
         \end{align}
We extend the map $x\mapsto x^*$ to  a conjugate-linear anti-automorphism of the algebra $\mathfrak{U}(\g g_\C^{})$. Thus $(D_1D_2)^*=D_2^*D_1^*$ for every $D_1,D_2\in\mathfrak{U}(\g g_\C^{})$. Such an extension is possible because of  $*$-invariance of $I$. 
We now define a  monoid
\[
\semS:=G_0\ltimes \mathfrak{U}(\g g_\C^{}),
\] with a multiplication given by
\[
(g_1,D_1)(g_2,D_2):=\big(g_1g_2,(\Ad(g_2^{-1})(D_1)) D_2\big).
\]
The neutral element of $\semS$ is $1_{\semS}:=(1_{G_0},1_{\mathfrak{U}(\mathfrak g_{\mathbb C})})$. The map
\[
\semS\to\semS\ ,\  
(g,D)  \mapsto (g,D)^*:=\big(g^{-1},\Ad(g)(D^*)\big)
\]
is an involution of $\semS$.

 The proof of the latter assertion is similar to the one in the $\Z_2$-graded case (see \cite{Koszul}, \cite{Kostant}, \cite[Thm 
5.5.2]{NeSaTG}).

 Next, for every $f\in C^\infty(\cG)$, we define a map 
\[        
\check{f}:\semS \to\C\ ,\ 
(g,D) \mapsto f(D)(g).
\]
Also, for any $a\in\Gamma$ we set 
$\semS_a:=\{(g,D)\in \semS\,:\,
|D|=a\}$.

\begin{dfn}
\label{dfn-posddd}
An $f\in C^\infty(\cG)$ is called \emph{positive definite} if it satisfies the following two conditions:
\begin{enumerate}
\item[(i)] $\check{f}(g,D)=0$ unless $(g,D)\in \semS_0$. 
\item[(ii)] $\sum_{1\leq i,j\leq n}\overline{c_i}c_j\check{f}(s_i^*s_j)\geq 0,\qquad 
\text{ for all } n\geq 1, c_1,\cdots,c_n\in \mathbb C, s_1,\cdots,s_n\in \semS$.
\end{enumerate}
\end{dfn}

Given a unitary representation $(\pi,\rho^\pi,\HH)$ of
$G$, for any two vectors $v,w\in\ccH$ we define the matrix coefficient $\varphi_{v,w}$ to be the map
\[
\varphi_{v,w}:\semS\to \C\ ,\ 
\varphi_{v,w}(g,D):=(\pi(g)\rho^\pi(D)v,w).
\]

\begin{prp}
\label{prp:v=wp}
Let $(\pi,\rho^\pi,\HH)$ be a smooth unitary representation of $(G_0,\mathfrak g)$, and let $v,w\in \HH^\infty$ be homogeneous vectors such that $|v|=|w|$. Then there exists an $f\in C^\infty(\cG)$ such that $\check{f}=\varphi_{v,w}$. Furthermore, $\check{f}(s)=0$ unless $s\in\semS_0$. If $v=w$ then $\check{f}$ is positive definite.
\end{prp}
\begin{proof}
Similar to \cite[Prop. 6.5.2]{NeSaTG}.
\end{proof}

We are now ready to describe the  GNS construction for unitary representations of $\Gamma$-Harish-Chandra pairs. It is an extension of the one  given in \cite[Sec. 6]{NeSaTG} in the $\Z_2$-graded case, and therefore we will skip the proof details.

Let $(\pi,\rho^\pi,\HH)$ be a smooth unitary representation of $(G_0,\mathfrak g)$. One can construct a 
\emph{$*$-representation}
$\widetilde{\rho^\pi}$ of the monoid $\semS$ by setting
\begin{align*}
\widetilde{\rho^\pi}:\semS \to 
\mathrm{End}_{\mathbb C}(\HH^\infty)\ ,\           (g,D) \mapsto \pi(g)\rho^\pi(D).
\end{align*}
Being a $*$-representation means that $\widetilde{\rho^\pi}(s^*)=
\widetilde{\rho^\pi}(s)^*$ for every $s\in \semS$, and in particular        
\begin{align}
\label{*representation}
(\widetilde{\rho^\pi}(s)v,w)=(v,\widetilde{\rho^\pi}(s^*)w).
\end{align}
By Proposition \ref{prp:v=wp}, for every matrix coefficient $\varphi_{v,v}$, where $v\in\ccH_0$, there exists a positive definite $f\in C^\infty(\cG)$ such that $\varphi_{v,v}=\check{f}$. 

Conversely, given a positive definite function $f\in C^\infty(\cG)$, one can associate a $*$-representation of $\semS$ to $f$ as follows. 
Set $\psi:=\check{f}$, and for every $s\in \semS$ let
$\psi_s:\semS\rightarrow \mathbb C$ be the map given by 
$\psi_s(t):=\psi(ts)$. We also set
\[\mathfrak D_{\psi}:=\spn_{\mathbb C}\left\{\psi_s\,:\,s\in \semS\right\}.
\]
Note that $\mathfrak D_\psi$ is a 
$\Gamma$-graded vector space of complex-valued functions on $\semS$, where the homogeneous parts of the $\Gamma$-grading are defined by
\[
\mathfrak D_{\psi,a}:=\{h\in\mathfrak D_\psi\,:\,
h(s)=0\text{ unless }s\in\semS_a\}.
\]
The space $\mathfrak D_\psi$
can be equipped with a sesquilinear form that is uniquely defined by the relation $(\psi_t,\psi_s):=\psi(s^*t)$. 
The completion $\HH_{\psi}$ of the resulting pre-Hilbert space is the \emph{reproducing kernel Hilbert space} that corresponds to the kernel 
\[     
K:\semS\times \semS \to \mathbb C\ ,\ (t,s)\mapsto \psi(ts^*).
\]
In other words, with respect to the inner product $(\cdot,\cdot)$ on $\HH_\psi$, we have
\[
h(s)=(h,K_s)\quad  \text{for }h\in\HH_{\psi},\ s\in \semS.
\]
There is a natural $*$-representation 
$\widetilde{\rho_\psi}:\semS \to\mathrm{End}_\C(\mathfrak D_\psi)$ by right translation, given by
\[
(\widetilde{\rho_\psi}(s)h)(t) :=h(ts)\quad\text{for }s,t\in S,\ h\in\mathfrak D_{\psi}.
\]
If $s\in\semS$ satisfies $ss^*=s^*s=1_\semS$, then $\widetilde{\rho_\psi}(s):\mathfrak D_{\psi}\rightarrow\mathfrak D_{\psi}$ is an isometry and extends uniquely to a unitary operator on $\HH_\psi$. Using the latter fact for elements of the form  $(g,1_{\mathfrak U(\g g_\C)})$ where $g\in G_0$, one obtains a unitary representation of $G_0$ on $\HH_\psi$ (in fact the vectors $K_s$, $s\in \semS$, have smooth $G_0$-orbits). 
Setting \[
\ccB:=\mathfrak D_\psi,\ \ccH:=\ccH_\psi,\ \rho^\ccB(x):=\widetilde{\rho_\psi}(1_{G_0},x)\text{ for }x\in\g g,
\text{ and } \pi(g):=\widetilde{\rho_\psi}(g,0),
\] we obtain a pre-representation of $(G_0,\g g)$.
Theorem \ref{stability theorem} implies that this pre-representation corresponds to a unique unitary representation of $(G_0,\g g)$.

\begin{dfn}
\label{cyclicvcc}
By a \emph{cyclic vector}
in a unitary representation $(\pi,\rho^\pi,\HH)$ we mean
a vector $v\in\HH$ such that $\widetilde{\rho^\pi}(\semS)v$ is dense in $\HH$. 
\end{dfn}

The above construction results in Theorem \ref{GNS} below.

\begin{theorem}
\label{GNS}
Let $(G_0,\g g)$ be a $\Gamma$-Harish-Chandra pair
such that
 \begin{equation}
 \label{eq-condgagbb}
\mathfrak g_{a}=\sum_{b,c\in\Gamma_\ood,bc=a}[\mathfrak g_{b},\mathfrak g_{c}].
\end{equation}
Also, let $f\in C^\infty(\cG)$ be  positive definite. 
\begin{enumerate}
\item[\rm (i)]
There exists a unitary representation $(\pi,\rho^\pi,\HH)$
of $(G_0,\g g)$ 
 with a cyclic vector $v_0\in\HH_0$ such that $\check{f}=\varphi_{v_0,v_0}$. 
\item[\rm (ii)] 
Let $(\sigma,\rho^\sigma,\ccK)$ be another unitary representation of $(G_0,\g g)$ with a cyclic vector $w_0\in\ccK_0$ such that   $\check{f}=\varphi_{w_0,w_0}$. Then $(\pi,\rho^\pi,\HH)$ and $(\sigma,\rho^\sigma,\ccK)$ are unitarily equivalent via an intertwining operator 
that maps $v_0$ to $w_0$.

\end{enumerate}

\end{theorem}
\begin{proof}
The proof  is an extension of the argument of \cite[Thm 6.7.5]{NeSaTG}.
\end{proof}

\end{document}